\documentclass[conference,letterpaper]{IEEEtran}
\IEEEoverridecommandlockouts

\usepackage[utf8]{inputenc}
\usepackage[T1]{fontenc}
\usepackage{cite}
\usepackage{url}
\usepackage{graphicx}
\usepackage{amsmath,amssymb}
\usepackage{bbm}
\usepackage{microtype}
\usepackage{glossaries}
\usepackage{bbm}

\newtheorem{proposition}{Proposition}

\newacronym{sbs}{SBS}{Stimulated Brillouin Scattering}
\newacronym{qkd}{QKD}{Quantum Key Distribution}

\begin{document}
\title{Fundamental Limits of Eavesdropper Detection in Optical Fiber via Stimulated Brillouin Scattering
%Information-Theoretic Limits of Eavesdropper Detection in Optical Fibers via Stimulated Brillouin Scattering
\thanks{The research is part of the Munich Quantum Valley, which is supported by the Bavarian state government with funds from the Hightech Agenda Bayern Plus. This work was financed by the DFG via grant NO 1129/2-1 and by the Federal Ministry of Education and Research of Germany in the Q-STARS project, grant number 16KIS2604, as well as via grants 16KISQ093, 16KISQ039 and 16KISQ077. The generous support of the state of Bavaria via the 6GQT project is greatly appreciated.}
}

\author{
\IEEEauthorblockN{Kiran Adhikari, Janis N\"otzel (Member, IEEE)}
\textit{Emmy-Noether Group Theoretical Quantum Systems Design},\\
\textit{Technical University of Munich, Munich, Germany}}
%\IEEEauthorblockA{
%Emmy Noether Group for Theoretical Quantum Systems Design\\
%Technical University of Munich, Germany}
%}

\maketitle

\begin{abstract}
 Recent work investigated the use of \gls{sbs} to measure changes in fiber parameters, thereby enhancing the security of a \gls{qkd} system. In this work, we focus solely on the impact of quantum technology on the task of intrusion-detection. We derive an effective input-output model for the \gls{sbs} interaction, and utilize it to compare three detection methods: First, the established state of the art. Second, a photon-counting based method which will likely be available in the near future and, finally, the ultimate quantum limit. We illustrate the potential benefit from modern quantum technology within two different mathematical frameworks: First by using the quantum error exponent of asymmetric hypothesis testing, and second in the context of parameter-estimation and quantum metrology.
\end{abstract}
\section{Introduction}

%$\mathbbm{i}$
Optical fiber networks are a critical component of modern communication infrastructure and are often assumed to provide strong physical-layer security. However, practical attacks such as evanescent coupling or fiber bending allow an adversary to extract optical power with minimal disturbance, making reliable detection of eavesdropping a fundamental challenge \cite{cca879e457e44e34b704c801c6aaa6e8, Shaneman2004OpticalNS}.  While prior work has studied eavesdropper detection and joint communication–sensing strategies in optical and quantum channels \cite{eavesdropperDetection, munarvallespirJCAS, boche2016, boche2017}, a quantitative understanding of the fundamental limits of detection in realistic fiber systems remains incomplete yet desirable when planning for adoption of modern quantum technology into future networks \cite{integratingQuantumSimulation}.

Stimulated Brillouin scattering (SBS) \cite{Wolff:21, Lu2019DistributedOF}, as depicted in Figure \ref{fig:brillouin}, is widely used for distributed sensing in optical fibers, as its resonance frequency is highly sensitive to local material properties such as strain and temperature. In SBS, a strong optical pump interacts with a counter-propagating Stokes field and an acoustic phonon mode, generating a narrowband Lorentzian gain response. The Brillouin resonance frequency is given by
\[
\frac{\Omega_B}{2\pi}=\frac{2}{\lambda}\,n_{\mathrm{eff}}(p,T)\,v_{\mathrm{ac}}(p,T),
\]
and can be measured by sweeping the pump–probe detuning and observing the resulting amplification. This enables spatially resolved sensing along the fiber and has been extensively exploited in techniques such as Brillouin optical correlation domain analysis (BOCDA) \cite{Lu2019DistributedOF, senaOTDR}.

\begin{figure}[h]
    \centering
    \includegraphics[width=0.9\linewidth]{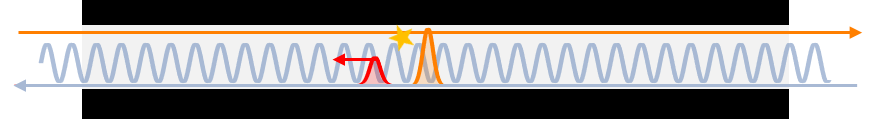}
    \caption{Light blue: CW light, depicting the probe, operating at frequency $\omega_s$. Orange: short, bright pulse, depicting the pump, operating at frequency $\omega_p$. The backscattered stokes photons are at the frequency $\omega_s$ of the probe. }
    \label{fig:brillouin}
\end{figure}

\gls{sbs} techniques were used to demonstrate experimentally that they allow for detecting and localizing eavesdropping in \gls{qkd}  \cite{Gisin_2002, popp2023eavesdropperlocalizationquantumclassical}. However as the focus of these prior works was on realizing detection and localization, no information-theoretic perspective allowing to assess the limits of the involved procedures has been conducted so far. In particular, it remained unclear how to quantify the detectability of weak tapping attacks and what fundamental limits govern such detection.

In this work, we take a first step towards such theoretical analysis and develop a quantum information-theoretic framework for \gls{sbs}-based intrusion detection in optical fibers. We show that, under standard approximations, each fiber segment can be modeled as a single-mode Gaussian quantum channel acting on the optical probe field, and that the full fiber corresponds to a cascade of such channels. This representation enables the application of quantum hypothesis testing tools to analyze the distinguishability between undisturbed and attacked channels.

The main contributions of this paper are as follows. In section \ref{sec:quantum_mechannical}, we derive a quantum-mechanical input–output model of \gls{sbs} and show that each fiber segment implements a Gaussian quantum channel. In section \ref{sec:fiber_cascade}, we model the entire fiber as a cascade of Gaussian channels, capturing both attenuation and SBS-induced noise. Finally, in section \ref{sec:information_theoretic_analysis}, we formulate the eavesdropper detection as an asymmetric quantum hypothesis testing problem and quantify performance using relative entropy and Quantum Fisher Information (QFI). We derive scaling laws for the minimum detectable attack strength and compare optimal and receiver-constrained detection strategies.

\section{Notation}
To shorten equations, we define for every $x\in[0,1]$ the quantity $x':=1-x$. The underlying Hilbert space in our analysis is the Fock space. Notation is borrowed from \cite{bookserfini} and \cite{HolevoBook}, so that $\hat a^\dagger$ and $\hat a$ are creation- and annihilation operators, and $S_N(\alpha):=D(\alpha)S_N(0)D(\alpha)^\dagger$ denotes the displaced thermal state with mean thermal photon number $N\geq0$ and displacement $\alpha\in\mathbb C$.

\section{Quantum Mechanical Picture}
\label{sec:quantum_mechannical}
We consider stimulated Brillouin scattering (SBS) in a one-dimensional traveling-wave geometry as a quantum interaction among three slowly varying bosonic fields: a strong optical pump field $a_p(z,t)$, a counter-propagating Stokes (or probe) optical field $a_s(z,t)$, and an acoustic phonon field $b(z,t)$. 
The interaction Hamiltonian \cite{zhu2024optoacousticentanglementcontinuousbrillouinactive, rakich2016quantumtheorycontinuumoptomechanics} takes the form

\begin{align}
H_{\mathrm{int}} =
\int dz\, \hbar g_0(z)
&(
a_p(z,t)\,a_s^\dagger(z,t)\,b^\dagger(z,t) \\ \nonumber
&+
a_p^\dagger(z,t)\,a_s(z,t)\,b(z,t))
\end{align}

Here, \(g_0(z)\) is an effective opto-acoustic coupling coefficient that may depend on position. The first term describes the annihilation of a pump photon together with the creation of a Stokes photon and a phonon, while the second term describes the reverse conversion process. Together, this ensures that the Hamiltonian is hermitian and that the evolution is unitary.

\subsection{Bi-linear Hamiltonian}

Assuming that the optical pump is sufficiently strong, it is well approximated by a classical coherent field rather than a fully dynamical quantum operator. In this regime, the pump envelope, $a_p(z,t)$, is replaced by its mean amplitude,
\begin{equation}
a_p(z,t) \approx \alpha_p(z,t)e^{-\mathrm{i}\omega_p t},
\end{equation}
with \(|\alpha_p| \gg 1\). Physically, this corresponds to a highly populated pump mode whose quantum fluctuations are negligible compared with its coherent amplitude. The pump, therefore, acts as a classical drive that mediates the interaction between the Stokes optical field and the acoustic phonon field.

Under this approximation, the original three-wave mixing Hamiltonian becomes effectively bilinear in the remaining quantum fields. It is convenient to define the pump-enhanced coupling
\begin{equation}
G(z,t) = g_0(z)\alpha_p(z,t),
\end{equation}
which combines the bare opto-acoustic coupling with the classical pump amplitude. The interaction Hamiltonian then takes the linearized form
\begin{align}
H_{\mathrm{lin}}
=
\int dz\,\hbar
&(
G(z,t)\,a_s^\dagger(z,t)b^\dagger(z,t) \\ \nonumber
&+
G^*(z,t)\,a_s(z,t)b(z,t))
\end{align}
up to rotating-frame conventions.

The term proportional to \(a_s^\dagger b^\dagger\) describes the simultaneous creation of a Stokes photon and an acoustic phonon, driven by the classical pump. Its Hermitian conjugate describes the reverse process, in which a Stokes photon and a phonon are annihilated. The linearized interaction therefore has the form of a two-mode squeezing Hamiltonian, with the pump serving as an external energy reservoir.

\subsection{Small-segment input--output map}
\label{sec:derivation_of_input_output}

Consider a fiber of length \(L\), divided into \(N\) segments of equal length \(\delta z = \frac{L}{N}.\)
Let \(z_i\) denote the center of the \(i\)th segment. Over a short segment of length \( \delta z\) centered at \(z_i\), the parameters may be treated as constant:
\begin{equation}
G(z)\approx G_i,\qquad
\Omega_B(z)\approx \Omega_{B,i},\qquad
\alpha_s(z)\approx \alpha_{s,i}. \nonumber
\end{equation}
The linearized \gls{sbs} interaction on the segment is
\begin{equation}
H_{\mathrm{lin},i}
=
\hbar\left(
G_i e^{\mathrm{i}\Delta k_i z} a_s^\dagger b^\dagger
+
G_i^* e^{-\mathrm{i}\Delta k_i z} a_s b
\right). \nonumber
\end{equation}
For a sufficiently short segment, the phase mismatch satisfies $\Delta k_i\,\delta z \ll 1,$ so the phase can be absorbed into the coupling $\widetilde G_i = G_i e^{\mathrm{i}\Delta k_i z_i}.$ The Hamiltonian then becomes
\begin{equation}
\label{eq:linearH}
H_{\mathrm{lin},i}
=
\hbar\left(
\widetilde G_i a_s^\dagger b^\dagger
+
\widetilde G_i^* a_s b
\right).
\end{equation}
Using the method layed out in \cite{zhu2024optoacousticentanglementcontinuousbrillouinactive, rakich2016quantumtheorycontinuumoptomechanics}, we can then obtain the explicit input--output relation of the Stokes photon mode as follows:%expressed in Proposition \ref{prop:input-ouput}.

\begin{proposition}
\label{prop:input-ouput}
The output Stokes operator satisfies the input--output relation
\begin{equation}
\label{eq:input_output}
a_{s,\mathrm{out}}(\Omega)
=
\mu_i(\Omega)\, a_{s,\mathrm{in}}(\Omega)
+
\sqrt{\eta'_i}\, v_i(\Omega)
+
\nu_i(\Omega)\, b_{\mathrm{in}}^\dagger(\Omega),
\end{equation}
where $\mu_i(\Omega)=\sqrt{\eta_i}+\kappa_i(\Omega)$, $\eta_i=e^{-\alpha_{s,i}\delta z}$, $\kappa_i(\Omega)=|\widetilde G_i|^2\chi_i(\Omega)\,\delta z$, $\nu_i(\Omega)=-\mathrm{i}\widetilde G_i\sqrt{\Gamma}\,\chi_i(\Omega)\,\delta z,$ and
\begin{equation}
\chi_i(\Omega)=\frac{1}{\Gamma/2+\mathrm{i}(\Omega-\Omega_{B,i})}.
\label{eq:chi_def}
\end{equation}
Here, \(\kappa_i(\Omega)\) is the frequency-dependent SBS-induced gain/loss coefficient,  \(\nu_i(\Omega)\) describes the added phonon noise coupled into the Stokes mode, $\chi_i(\Omega)$ is the local Brillouin susceptibility, and $\sqrt{\eta_i}$ describes pure attenuation over a length $\delta z$.
\end{proposition}

Therefore, each segment implements a bosonic Gaussian channel, consisting of a deterministic linear transformation, vacuum noise associated with optical attenuation, and thermal noise from the phonon bath. Since evolution is generated by a quadratic Hamiltonian and linear coupling to Gaussian environments, the map is Gaussian and preserves the state's Gaussian character.

\section{Continuous variable formalism}
\label{sec:fiber_cascade}
Continuous-variable (CV) quantum systems are bosonic modes described by canonical quadrature operators rather than finite-dimensional qubits \cite{bookserfini, HolevoBook}. For an $n$-mode system, we define the quadrature vector $\hat{\mathbf R}
=
(\hat x_1,\hat p_1,\ldots,\hat x_n,\hat p_n)^T,$ with canonical commutation relations $[\hat R_j,\hat R_k]= \mathrm{i}\,\tilde \Omega_{jk},$
where $\tilde \Omega$ is the symplectic form. For a single mode with annihilation operator $\hat a$, the quadratures are
\begin{equation}
q=\frac{ a+ a^\dagger}{\sqrt 2},
\qquad
 p=\frac{ a-  a^\dagger}{\mathrm{i}\sqrt 2}.
\end{equation}

A Gaussian state is any state whose Wigner function is Gaussian in phase space. Such a state is completely characterized by its first and second moments. The mean vector is $ d=\langle \hat{R}\rangle$
and the covariance matrix is $V_{jk}
=
\frac12\left\langle
\{\hat R_j-d_j,\hat R_k-d_k\}
\right\rangle.$

We will now show that each fiber segment $i$ under the linearized stimulated Brillouin scattering interaction defines a phase-insensitive single-mode Gaussian channel acting on the Stokes mode. For this, we consider an optical probe mode entering segment \(i\) as modeled by a local Gaussian channel \(\mathcal E_i(z_c)\), whose parameters depend on the correlation-peak position \(z_c\). As derived in Section \ref{sec:derivation_of_input_output}, the input Stokes mode in the segment $i$ evolves as 
\begin{equation}
\begin{aligned}
\label{eq:aout_ain}
a_{s,\mathrm{out}}(\Omega)
&=
\mu_i(\Omega)\,a_{s,\mathrm{in}}(\Omega)
+
\sqrt{\eta'_i}\,v_i(\Omega)
+
\nu_i(\Omega)\,b_{\mathrm{in}}^\dagger(\Omega), \\
a_{s,\mathrm{out}}^\dagger(\Omega)
&=
\mu_i^*(\Omega)\, a_{s,\mathrm{in}}^\dagger(\Omega)
+
\sqrt{\eta'_i}\, v_i^\dagger(\Omega)
+
\nu_i^*(\Omega)\, b_{\mathrm{in}}(\Omega).
\end{aligned}
\end{equation}
where $\mu_i(\Omega)=\sqrt{\eta_i}+\kappa_i(\Omega)$. Therefore, each fiber segment can be modeled as a single-mode Gaussian channel.

\begin{proposition}[Gaussian channel representation of a fiber segment]
Each fiber segment $i$ under the linearized stimulated Brillouin interaction implements a single-mode Gaussian quantum channel acting on the Stokes mode. In quadrature form, the input--output relations are
\begin{align}
d_{\mathrm{out}} &= X_i d_{\mathrm{in}}, \\
V_{\mathrm{out}} &= X_i V_{\mathrm{in}} X_i^T + Y_i,
\end{align}
where the deterministic matrix is
\begin{equation}
X_i(\Omega)=
\begin{pmatrix}
\sqrt{\eta_i}+\Re \kappa_i(\Omega) & -\Im \kappa_i(\Omega) \\
\Im \kappa_i(\Omega) & \sqrt{\eta_i}+\Re \kappa_i(\Omega)
\end{pmatrix},
\end{equation}
and the added-noise matrix is
\begin{equation}
Y_i=
\left[
\left(n_{\mathrm{th},i}+\frac{1}{2}\right)|\nu_i(\Omega)|^2
+
\frac{\eta'_i}{2}
\right] I_2.
\end{equation}
where
\begin{equation}
n_{\mathrm{th},i}
=
\frac{1}{\exp\!\left(\frac{\hbar \Omega_{B,i}}{k_B T_m}\right) - 1}
\end{equation}
denotes the thermal photon occupation of the acoustic field at temperature $T_m$.

\end{proposition}
This shows that $\sqrt{\eta_i}$ represents attenuation, $\Re \kappa_i(\Omega)$ gives \gls{sbs} gain/loss, $\Im \kappa_i(\Omega)$ produces phase rotation, and $Y_i$ contains both vacuum noise and thermal Brillouin noise.
\begin{IEEEproof}
Starting from the input--output relation
\begin{equation}
\hat a_{s,\mathrm{out}} = \mu_i \hat a_{s,\mathrm{in}} + \sqrt{\eta'_i}\,\hat v_i + \nu_i \hat b_i^\dagger,
\end{equation}
we express the operators in terms of quadratures using
\[
q = \frac{a + a^\dagger}{\sqrt{2}}, \qquad
p = \frac{a - a^\dagger}{\mathrm{i}\sqrt{2}}.
\]
Writing $\mu_i = \mu_{R,i} + \mathrm{i}\mu_{I,i}$ and $\nu_i = \nu_{R,i} + \mathrm{i}\nu_{I,i}$, the transformation becomes linear in the quadrature operators:
\begin{equation}
R_{\mathrm{out}} = X_i R_{\mathrm{in}} + Z_i R_b + \sqrt{\eta'_i}\, R_v,
\end{equation}
where $R=(q,p)^T$, and $Z_i$ is the corresponding phonon-coupling matrix. Explicitly

\begin{equation}
X_i=
\begin{pmatrix}
\Re \mu_i & -\Im \mu_i \\
\Im \mu_i & \Re \mu_i
\end{pmatrix},
\quad
Z_i=
\begin{pmatrix}
\Re \nu_i & \Im \nu_i \\
\Im \nu_i & -\Re \nu_i
\end{pmatrix}.
\label{eq:XiZi}
\end{equation}

Taking expectation values yields $d_{\mathrm{out}} = X_i d_{\mathrm{in}}$, since the bath modes satisfy $d_b = d_v = 0$.

For the covariance matrix, using the independence of input, vacuum, and phonon modes,
\begin{equation}
V_{\mathrm{out}} = X_i V_{\mathrm{in}} X_i^T + Z_i V_b Z_i^T + \eta'_i\cdot V_v.
\end{equation}
For a thermal phonon bath and vacuum optical bath,
\begin{equation}
V_b = \left(n_{\mathrm{th},i}+\frac{1}{2}\right) I_2,
\qquad
V_v = \frac{1}{2} I_2.
\end{equation}
 Substituting these gives
\begin{equation}
\begin{aligned}
Y_i &= Z_i V_b Z_i^T + \eta'_i\cdot V_v \nonumber \\
&=
\left[
\left(n_{\mathrm{th},i}+\tfrac{1}{2}\right)|\nu_i|^2
+
\tfrac{\eta'_i}{2}
\right] I_2.
\end{aligned}
\end{equation}
 Thus, each segment acts as a rotation-scaling transformation followed by isotropic Gaussian noise, i.e.\ a phase-insensitive Gaussian channel.
\end{IEEEproof}

\subsection{Fiber as a cascade of quantum channels}

The entire fiber is therefore described by a cascade of quantum channels,
\begin{equation}
\mathcal E_{\mathrm{fiber}}
=
\mathcal E_L
\circ
\mathcal E_{L-1}
\circ
\cdots
\circ
\mathcal E_1.
\label{eq:fiber_cascade}
\end{equation} 
The cascade of two Gaussian channels with parameters \((X_1,Y_1)\) and \((X_2,Y_2)\) is again Gaussian, with total deterministic matrix $X_{21}=X_2X_1$, and the noise matrix $Y_{21}=X_2Y_1X_2^T+Y_2.$ 

Hence, the entire fiber, modeled as a cascade of \(L\) local Gaussian channels, has total parameters
\begin{equation}
\begin{aligned}
\label{eqn:fiber-channel}
X_{\mathrm{tot}} &=X_LX_{L-1}\cdots X_1, \\
Y_{\mathrm{tot}}
&=
\sum_{k=1}^L
\left(
X_LX_{L-1}\cdots X_{k+1}
\right)
Y_k
\left(
X_LX_{L-1}\cdots X_{k+1}
\right)^T.
\end{aligned}
\end{equation}

Using the fact that each \(X_i\) has the rotation-scaling form and each \(Y_k\) is proportional to \(\mathbb I_2\), the explicit form of the parameters $X_{\mathrm{tot}}$ and $Y_{\mathrm{tot}}$  are as follows:

\begin{equation}
  X_{\mathrm{tot}}=
\begin{pmatrix}
\Re \mu_{\mathrm{tot}} & -\Im \mu_{\mathrm{tot}}\\
\Im \mu_{\mathrm{tot}} & \Re \mu_{\mathrm{tot}}
\end{pmatrix},
\qquad
Y_{\mathrm{tot}}=y_{\mathrm{tot}}\mathbb I_2
\end{equation}
where parameters $\mu_{\mathrm{tot}}$ and $ y_{\mathrm{tot}}$ given by
\begin{equation}
    \begin{aligned}
        \mu_{\mathrm{tot}}&=\prod_{k=1}^L \mu_k \\
        y_{\mathrm{tot}} &=\sum_{k=1}^N
\left(
\prod_{j=k+1}^N |\mu_j|^2
\right)
\left[
\tfrac{\eta'_k}{2}
+
\left(n_{\mathrm{th},k}+\tfrac12\right)|\nu_k|^2
\right].
    \end{aligned}
\end{equation}

\section{Information-Theoretic Analysis}
\label{sec:information_theoretic_analysis}

Given the channel \eqref{eqn:fiber-channel}, we formulate the task of detecting whether, at a given segment $i$, the channel parameters deviate from their nominal values due to an eavesdropping attack. The input states are coherent states, which are transformed by the Gaussian channel $(X,Y)$ into displaced thermal states.

\subsection{Channel Model and Attack Description}

We assume uniform optical loss \(\eta_k = \eta\) for all segments and a coherent input probe state with amplitude \(\sqrt{E}\). In the absence of pump interaction, $\widetilde G_k = 0$ for all $k$, so that \(\mu_k = \sqrt{\eta}\) and no \gls{sbs} noise is generated. The resulting channel is a pure-loss Gaussian channel with output state \(S_O(\gamma)\), where
$O = \frac{1 - \eta^L}{2}$ and $\gamma = \eta^{L/2} \sqrt{E}$.

We consider network operators aiming to ensure physical-layer security by detecting any eavesdropper with high probability, while allowing occasional false alarms. This motivates an asymmetric hypothesis-testing formulation. Assuming $k$ independent probe events and ideal initial conditions across all segments $i=1,\ldots,L$, all segments can be monitored in parallel. The detection task at segment $i$ is therefore

\begin{align}
\min_{0 \le P \le \mathbb{I}} 
\left\{
\operatorname{Tr}\bigl(P S_N(\beta)^{\otimes k}\bigr)
:
\operatorname{Tr}\bigl((\mathbb{I}-P) S_M(\alpha)^{\otimes k}\bigr) \ge \epsilon'
\right\},
\end{align}

where $(\alpha, M)$ and $(\beta, N)$ describe the clean and disturbed states, respectively.

\subsubsection{Clean Channel}

When the pump overlaps segment $i$, \gls{sbs} interaction is activated locally. The effective gain coefficient becomes $\mu_{\mathrm{clean}}^{(i)} = \eta^{(L-1)/2} (\sqrt{\eta} + \kappa_i).$
The output state is $S_M(\alpha)$ with parameters
\begin{align}
\label{eq:clean_parameters}
M &= \frac{1 - \eta^L}{2}
+ \eta^{L-i} \left(n_{\mathrm{th}} + \tfrac{1}{2}\right) |\nu_i|^2, \\
\alpha &= \eta^{(L-1)/2} (\sqrt{\eta} + \kappa_i) \sqrt{E}.
\end{align}

\subsubsection{Disturbed Channel}

An evanescent-coupling attack is modeled as a beam-splitter interaction at segment $i$ with transmissivity $\tau_E$. The output state is $S_N(\beta)$ with
\begin{align}
\label{eq:disturbed_parameters}
N &= \tfrac{1 - \eta^L}{2}
+ \eta^{L-i}
\left[
\tau_E \left(n_{\mathrm{th}} + \tfrac{1}{2}\right) |\nu_i|^2
+ \tau_E'\left(n_E + \tfrac{1}{2}\right)
\right], \\
\beta &= \eta^{(L-1)/2} \sqrt{\tau_E} (\sqrt{\eta} + \kappa_i) \sqrt{E}.
\end{align}

Here, $1-\tau_E$ is the extracted power fraction. Passive attacks ($n_E=0$) introduce only vacuum noise, while active attacks ($n_E>0$) introduce excess noise.

\subsection{Relative-Entropy Analysis}

The relative entropy between displaced thermal Gaussian states  can be calculated from \cite[Lemma 3]{parthasarathy2021pedagogicalnotecomputationrelative} and \cite{HolevoBook} as 
\begin{align}
D(S_N(\beta)\|S_M(\alpha)) 
&= -g(N) + \ln(M+1) \nonumber \\
&\quad + \ln\!\left(\tfrac{M+1}{M}\right)\left(N + |\beta - \alpha|^2\right),
\label{eqn:relative-entropy-formula}
\end{align}
where $g(x) = (x+1)\ln(x+1) - x\ln x$. Substituting the channel parameters yields
\begin{align}
D
&= -g(N) + \ln(M+1) \nonumber \\
&\quad + \ln\!\left(\tfrac{M+1}{M}\right)
\left[
N + \eta^{L-1} |\sqrt{\tau_E} - 1|^2
|\sqrt{\eta} + \kappa_i|^2 E
\right].
\label{eq:quantum_Stein}
\end{align}
Before we perform the numerical analysis, let us consider a weak attack approximation, which yields rough analytical results. In this limit, $\tau_E = 1-\rho$ such that $\rho \ll 1$ and $\sqrt{\tau_E} \approx 1 - \frac{\rho}{2}$. Expanding the relative entropy to the leading order in $\rho$, one finds
\begin{align}
D \equiv & D\!\left(S_N(\beta)\,\|\,S_M(\alpha)\right)
= 
\ln\!\left(\frac{M+1}{M}\right)\,|\beta-\alpha|^2  \\
&=\frac{\rho^2}{4}\,
\ln\!\left(\frac{M+1}{M}\right)\,
\eta^{L-1}|\sqrt{\eta}+\kappa_i|^2 E. \nonumber
\end{align}

Hence, after $k$ independent probes, using the quantum Stein lemma, $D_k = kD.$ By the quantum Stein lemma, the missed-detection probability decays as $P_{\mathrm{miss}} \sim e^{-kD}.$ Requiring $P_{\mathrm{miss}} \lesssim p$ yields the
detectability condition $kD \gtrsim \ln\!\frac{1}{p}.$ Solving for $\rho$, we obtain that the minimum detectable tap strength is of the order
\begin{equation}
\rho_{\min}
\approx
2\sqrt{
\frac{\ln(1/p)}
{k\,\ln\!\left(\frac{M+1}{M}\right)\,\eta^{L-1}\,|\sqrt{\eta}+\kappa_i|^2\,E}
}.
\end{equation}
Thus, the weakest detectable evanescent-coupling attack scales as
$\rho_{\min}\sim (kE)^{-1/2}$, with additional suppression from
fiber loss and enhancement from the local \gls{sbs} response.

Similarly, the security condition  $2^{-k D(S_M(\alpha)\|S_N(\beta))} \le \lambda$ is equivalent to $k D \gtrsim \ln\!\frac{1}{\lambda}.$ This yields the minimum number of probes as
\begin{equation}
k \gtrsim
\frac{4\,\ln(1/\lambda)}
{\rho^2\,\ln\!\left(\frac{M+1}{M}\right)\,
\eta^{L-1}|\sqrt{\eta}+\kappa_i|^2 E}.
\end{equation}
Suppose  $t_L$ is the time a pulse takes to run through a fiber of length $L$. The total scan time is given by \(T_{\mathrm{scan}} = t_L \cdot k\). This number provides the time window that the eavesdropper can use at best, and therefore quantifies the maximum amount of information he could get before data transmission stops, as $ B_{stolen} (\rho) = T_{scan}\cdot C (\rho)$ where $C$ is the data transmission capacity of the WDM system in case that it is run in parallel with the test (Brillouin scattering) system. Hence,
\begin{equation}
B_{\mathrm{stolen}}(\rho)
\approx
C\,t_L\,
\frac{4\ln(1/\lambda)}
{\rho^2\,
\ln\!\left(\frac{M+1}{M}\right)\,
\eta^{L-1}\,
|\sqrt{\eta}+\kappa_i|^2\,E},
\end{equation}
which shows the characteristic scaling $B_{\mathrm{stolen}} \propto \rho^{-2}.$

\subsubsection{Numerical analysis}
In the numerical analysis part, in addition to the quantum Stein exponent $D \equiv D(S_N(\beta)\|S_M(\alpha))$ obtained in Eq. \eqref{eq:quantum_Stein}, we evaluate two other receiver-constrained exponents.

The first is photon counting with a threshold decision rule based on the total detected photon number, where the exponent $D_P$ is given by
\begin{align}
D_P &=  \sup_{s\le 0}\Big[ sM + \ln\!\bigl(N+1 - N e^{s}\bigr)
%\nonumber\\
%&\qquad
+ \Delta\,
\tfrac{1 - e^{s}}{N+1 - N e^{s}}
\Big].
\end{align}
In our model,
\begin{equation}
\Delta = |\beta-\alpha|^2
=
\eta^{L-1}|\sqrt{\tau_E}-1|^2\,|\sqrt{\eta}+\kappa_i|^2\,E,
\end{equation}
so that $D_P$ is fully determined by the channel parameters $M,N$ and the displacement $\Delta$.

The second is heterodyne detection, where we define $D_H$ as the classical relative entropy between the corresponding heterodyne outcome distributions \cite[Chapter 1]{pardo2005statistical}, obtained as follows:

\begin{equation}
D_H
\equiv
\int_{\mathbb C} d^2 z \, q(z|N,\beta)\,
\ln\!\tfrac{q(z|N,\beta)}{q(z|M,\alpha)},
\end{equation}
where $q(z|\bar n,\gamma)
=
\frac{1}{\pi(\bar n+1)}
\exp\!\left(-\tfrac{|z-\gamma|^2}{\bar n+1}\right).$ Evaluating the Gaussian integral gives
\begin{equation}
D_H
=
\ln\!\left(\tfrac{M+1}{N+1}\right)
+\tfrac{N+1}{M+1}-1
+\tfrac{|\beta-\alpha|^2}{M+1}.
\end{equation}

\begin{figure}[h]
    \centering
    \includegraphics[width=0.9\linewidth]{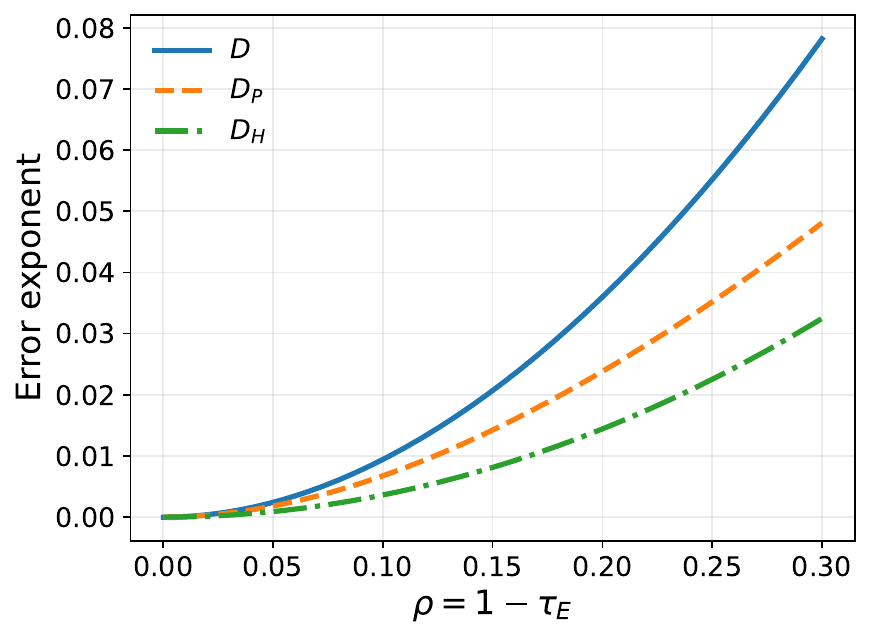}
    \caption{Comparison of the quantum Stein exponent $D$, photon number threshold detection exponent $D_P$, and heterodyne exponent $D_H$ for the clean and disturbed channel model as a function of the attack strength $\rho = 1-\tau_E$. Further parameters: $\eta = 0.98$, $L = 100$, attacked segment index: $50$, $E=5.0$, $n_{\mathrm{th}}=1.0$, $n_E=0.5$.}
    \label{fig:exponent_vs_attack_sstrength}
\end{figure}

\begin{figure}[h]
    \centering
    \includegraphics[width=0.9\linewidth]{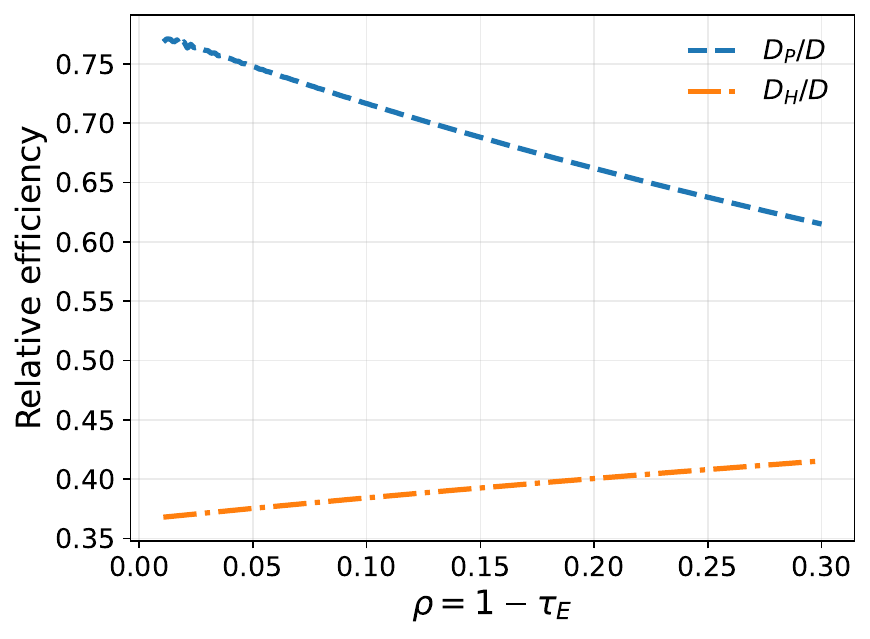}
   \caption{
Relative efficiency of photon-number threshold detection and heterodyne detection compared to the optimal quantum receiver, shown as $D_P/D$ and $D_H/D$ as a function of the attack strength $\rho = 1-\tau_E$. Further parameters: $\eta = 0.98$, $L = 100$, attacked segment index: $50$, probe energy $E=5.0$, $n_{\mathrm{th}}=1.0$, $n_E=0.5$.
}
    \label{fig:exponent_ratios}
\end{figure}

All three error exponents as a function of attack strength are plotted in Fig. \ref{fig:exponent_vs_attack_sstrength} and their corresponding ratios in Fig. \ref{fig:exponent_ratios}. As shown in Fig.~\ref{fig:exponent_vs_attack_sstrength}, all three error exponents increase monotonically with $\rho$, reflecting the increasing distinguishability between the clean and disturbed channels as the tapping strength increases. The quantum Stein exponent $D$ provides the optimal performance, while the receiver-constrained exponents satisfy $D_P > D_H$ across the entire range, indicating that photon-number threshold detection is more effective than heterodyne detection at capturing the relevant statistical differences induced by the attack. This behavior suggests that the disturbance is not purely a coherent displacement, but also modifies photon-number statistics in a way that favors photon-counting strategies. Figure~\ref{fig:exponent_ratios} complements this by quantifying the performance relative to the optimal limit: the photon threshold detector achieves approximately $60$--$80\%$ of the optimal exponent, whereas heterodyne detection remains limited to about $35$--$40\%$.

\subsection{Quantum Fisher Information via Quantum Sensing}%and Quantum-Enhanced Probing}

In addition to the hypothesis-testing formulation above, we can also analyze the intrusion-detection problem from a local
parameter-estimation perspective. We treat the attack strength
\(\rho = 1 - \tau_E\) as an unknown scalar parameter encoded in the
Gaussian output state and quantify its estimability via the quantum
Fisher information (QFI). 

For the coherent probe considered above, the output state is $ S_{N(\rho)}\bigl(\beta(\rho)\bigr)$ from Eq. \ref{eq:disturbed_parameters}. The covariance matrix for $ S_{N(\rho)}\bigl(\beta(\rho)\bigr)$ is
isotropic, $\Sigma(\rho)
=
\left(N(\rho)+\frac{1}{2}\right) I_2,$ and the displacement vector is $\Delta X(\rho)
=
\sqrt{2}
\begin{pmatrix}
\mathrm{Re}\,\beta(\rho)\\
\mathrm{Im}\,\beta(\rho)
\end{pmatrix}.$ The purity is $P(\rho)
=
\frac{1}{2N(\rho)+1}.$

For $k$ independent probe events, the quantum Cramér--Rao bound yields $    \mathrm{Var}(\hat{\rho})
    \geq
    \frac{1}{k F_Q(\rho)}.$
where, $F_Q(\rho)$ is the QFI. The QFI for a Gaussian state is  \cite{Pinel_2013, monras2013phasespaceformalismquantum}

\begin{equation}
\label{eq:general_form}
F_Q(\rho)
=
\frac12
\frac{\mathrm{tr}\!\left[(\Sigma^{-1}\dot{\Sigma})^2\right]}{1+P^2}
+
2\frac{\dot P^2}{1-P^4}
+
\Delta \dot X^T \Sigma^{-1} \Delta \dot X
\end{equation}
where $\dot{(\cdot)} \equiv \partial_\rho (\cdot)$. 
For a displaced thermal state, the QFI is then
\begin{equation}
    F_Q^{\rm coh}(\rho)
    =
    \frac{
    4|\partial_\rho \beta(\rho)|^2
    }{
    2N(\rho)+1
    }
    +
    \frac{
    [\partial_\rho N(\rho)]^2
    }{
    N(\rho)[N(\rho)+1]
    } .
\end{equation}

Similarly, defining the centered quadrature operator
\begin{equation}
    \Delta \hat X(\rho)
    =
    \hat X-\Delta X(\rho),
    \qquad
    \hat X=
    \begin{pmatrix}
        \hat q\\
        \hat p
    \end{pmatrix}.
\end{equation}
and let $ \hat b=\hat a-\beta(\rho)$ to denote thermal state in the displaced frame, we obtain symmetric logarithmic derivative (SLD) as:
\begin{align}
    \mathcal L_\rho
    & = \mathcal L_{\rm disp}
    +
    \mathcal L_{\rm noise}  \\ &=
    \Delta \dot X^{T}\Sigma^{-1}\Delta \hat X
    +
    \tfrac{\partial_\rho N}{N(N+1)}
    \left[
    (\hat a-\beta)^\dagger(\hat a-\beta)-N
    \right].\nonumber
\label{eq:sld_displaced_thermal}
\end{align}
where the first term is linear in quadrature, while the second is quadratic in field operators, i.e., number-like. The quantum Cramér--Rao bound is saturated by measurements in the
eigenbasis of the symmetric logarithmic derivative (SLD). Therefore, the
optimal measurement is given by the spectral decomposition of
\(\mathcal L_\rho\). The optimal POVM is
\begin{equation}
    \Pi_n
    =
    D\!\left(-\tfrac{A}{\lambda}\right)
    |n\rangle\langle n|
    D^\dagger\!\left(-\tfrac{A}{\lambda}\right),
\end{equation}
where
\begin{equation}
    A
    =
    \frac{2\partial_\rho\beta(\rho)}{2N(\rho)+1},
    \qquad
    \lambda
    =
    2\frac{\partial_\rho N(\rho)}{N(\rho)[N(\rho)+1]}.
\end{equation}
In Figure \ref{fig:fisher_information}, we have plotted the estimation variance as a function of the attack strength for different measurement strategies for the weak signal regime. 
\begin{figure}[h]
    \centering
    \includegraphics[width=\linewidth]{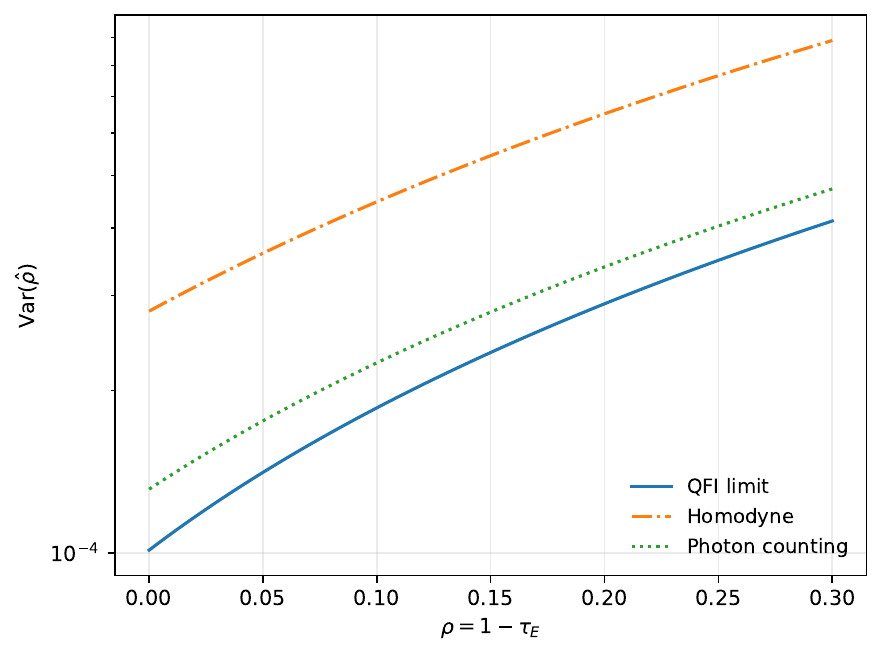}
    \caption{
Variance of the estimator $\hat{\rho}$ for various measurement methods as a function of the attack strength $\rho = 1 - \tau_E$. Further parameters: $\eta = 0.98$, $L= 100$, attacked segment index: $50$, $E = 0.5$, $n_{\mathrm{th}} = 1.0$, $n_E = 2.0$.
}
    \label{fig:fisher_information}
\end{figure}

It is crucial to note that the optimal POVM depends on the parameter $\rho$, which we are trying to estimate. To tackle this issue, one can use the Quantum Local Asymptotic Normality (QLAN) approach \cite{Yamagata_2013}, which implements an adaptive, near-optimal measurement that asymptotically saturates the QFI.

\section{Conclusion}

In this work, we modeled the fiber as a cascade of Gaussian channels, enabling us to develop a quantum information-theoretic framework for detecting eavesdropping in optical fibers via stimulated Brillouin scattering (SBS). Our results show that the detectability of weak tapping attacks scales with probe energy and the number of probing events, while being suppressed by propagation loss, leading to a characteristic inverse square-root scaling of the minimum detectable attack strength. In addition to the hypothesis-testing formulation, we also analyzed the intrusion-detection problem from a local parameter-estimation perspective where the attack strength is an unknown scalar parameter.

Further work is needed to sharpen the information-theoretic definition of the task and to assess the suitability of different probe signals. Possible methods include quickest change point detection \cite{Fanizza_2023}, compound hypothesis testing approaches \cite{Lami_2025} or multi-parameter metrology settings, where displacement and noise are treated as independent parameters. The latter would require the use of the Holevo Cramér–Rao bound instead of the QFI.

\bibliographystyle{IEEEtran}
\bibliography{bib}

%\newpage
%\section{Appendix}
%\input{Extras/displacedPhotonCounterOptimality}
\end{document}